\documentclass[11pt]{article}

\usepackage{amssymb} 
\usepackage{amsmath}     
\usepackage{amsthm}

 \usepackage{mathrsfs}
 \usepackage{psfrag}

\usepackage{a4wide}

\newtheorem{thm}{Theorem}

\newtheorem{cor}{Corollary}

\newenvironment{keyword}{\par{\noindent\bf Keywords:}}

\begin{document}

\title{On the approximability of minmax (regret) network optimization problems}

\author{Adam Kasperski\thanks{Corresponding author}\\
   {\small \textit{Institute of Industrial}}
  {\small \textit{Engineering and Management,}}
  {\small \textit{Wroc{\l}aw University of Technology,}}\\
  {\small \textit{Wybrze{\.z}e Wyspia{\'n}skiego 27,}}
  {\small \textit{50-370 Wroc{\l}aw, Poland,}}
  {\small \textit{adam.kasperski@pwr.wroc.pl}}
  \and
  Pawe{\l} Zieli{\'n}ski\\
    {\small \textit{Institute of Mathematics}}
  {\small \textit{and Computer Science}}
  {\small \textit{Wroc{\l}aw University of Technology,}}\\
  {\small \textit{Wybrze{\.z}e Wyspia{\'n}skiego 27,}}
  {\small \textit{50-370 Wroc{\l}aw, Poland,}}
  {\small \textit{pawel.zielinski@pwr.wroc.pl}}}

\date{}
\maketitle

\begin{abstract}
  In this paper the minmax (regret) versions of some basic polynomially solvable deterministic network problems are discussed. It is shown that if the number of scenarios is unbounded, then the problems under consideration are not approximable within $\log^{1-\epsilon} K$ for any $\epsilon>0$ unless NP $\subseteq$ DTIME$(n^{\mathrm{poly} \log n})$, where $K$  is the number of scenarios.
\end{abstract}

\begin{keyword}
  Combinatorial optimization;
  Approximation;
  Minmax;
  Minmax regret;
\end{keyword}

\section{Introduction}

We are given a network modeled by a directed or undirected graph $G=(V,E)$ with nonnegative cost $c_e$ associated with every edge $e\in E$. A set of feasible solutions $\Phi$ consists of some subsets of the edges of $G$. It may contain, for instance, all $s-t$ paths, spanning trees, $s-t$ cuts or matchings in $G$. 
In a classical \emph{deterministic network problem} $\mathcal{P}$, i.e. the problem in which the costs $c_e$ are precisely given, 
we wish to find a feasible solution $X\in\Phi$ that minimizes the total cost, namely the value of $\sum_{e\in X} c_e$. In this paper we assume that $\mathcal{P}$ is polynomially solvable. A comprehensive review of various polynomially solvable deterministic 
network problems~$\mathcal{P}$ can be found in~\cite{AH93, law76}.

In practice, the costs $c_e$  in the objective may be  uncertain.
In \textit{robust
approach}~\cite{KO97} the uncertainty is modeled by specifying a
set $\Gamma$ that contains all possible realizations of the edge costs. Every
particular realization $S=(c_e^S)_{e\in E}$ is called a \emph{scenario} and the value of $c_e^S$ denotes
the cost of edge $e$ under scenario~$S$.   There are two
ways of describing the set~$\Gamma$. 
In the \emph{interval scenario case},
the value of every edge cost may fall within a given closed
interval and~$\Gamma$ is the Cartesian product of
all the uncertainty intervals.
In the \emph{discrete
scenario case}, which is considered in this paper, 
the set of scenarios is defined by explicitly listing all scenarios.
So, $\Gamma=\{S_1,\dots,S_K\}$ is finite and contains exactly $K$ scenarios.
The cost of solution $X\in \Phi$ under scenario $S\in \Gamma$ is $F(X,S)=\sum_{e\in X} c^{S}_e$.
We will denote by 
$F^*(S)=\min_{X\in \Phi} F(X,S)$ the cost of an 
optimal solution under~$S$. 
In the \emph{minmax} version of problem~$\mathcal{P}$, we seek a solution that minimizes the worst case objective value over all scenarios, that is
\[
\textsc{Minmax}~\mathcal{P}:\; \min_{X\in \Phi}\max_{S\in\Gamma} F(X,S).
\]
In the \emph{minmax regret}  version of problem $\mathcal{P}$, we wish to find a solution that minimizes the \emph{maximal regret}, that is
\[
\textsc{Minmax Regret}~\mathcal{P}:\; 
\min_{X\in \Phi} Z(X)=\min_{X\in \Phi} \max_{S\in \Gamma}\{F(X,S)-F^*(S)\}.
\]
 The motivation of the minmax (regret) approach and a deeper discussion on the two robust criteria can be found in~\cite{KO97}. Unfortunately, under the discrete scenario case, the minmax (regret) versions of basic network problems such as \textsc{Shortest Path}, \textsc{Minimum Spanning Tree}, \textsc{Minimum Assignment} and \textsc{Minimum s-t Cut} turned out to be NP-hard even if~$\Gamma$ contains only~2 scenarios~\cite{AI05a, AI05, KO97}. Furthermore, if the number of scenarios is unbounded (it is a part of the input), then \textsc{Minmax (Regret) Shortest Path} is strongly NP-hard and not approximable within~$(2-\epsilon)$ and \textsc{Minmax (Regret) Minimum Spanning Tree} is strongly NP-hard and not approximable within~$(3/2-\epsilon)$ for any $\epsilon>0$ if P$\neq$NP~\cite{AI05c,AI07}. For the interval scenario case, if $\mathcal{P}$ is polynomially solvable, then \textsc{Minmax}~$\mathcal{P}$ is polynomially solvable as well. On the other hand, the minmax regret versions of \textsc{Shortest Path}, \textsc{Minimum Spanning Tree}, \textsc{Minimum Assignment} and \textsc{Minimum s-t Cut} are strongly NP-hard~\cite{AI05a,AI05, ARO04, AV04, PZ04}. It is worth pointing out that there are some interesting differences between the discrete and interval scenario representations. In~\cite{AV01} a minmax regret problem has been described, which is polynomially solvable in the interval case and  NP-hard for two explicitly given scenarios. On the other hand, the minmax regret linear programming problem is polynomially solvable in the discrete scenario case and becomes strongly NP-hard in the interval one~\cite{AV05}.

 Consider again the discrete scenario case. If problem~$\mathcal{P}$ is polynomially solvable and the edge costs under all scenarios are nonnegative, then both  \textsc{Minmax}~$\mathcal{P}$ and \textsc{Minmax Regret}~$\mathcal{P}$ are approximable within $K$~\cite{AI06}.
 A~generic  $K$-approximation algorithm proposed in~\cite{AI06} simply outputs an optimal solution to problem~$\mathcal{P}$ under costs $c_e=\frac{1}{K}\sum_{S\in \Gamma} c_e^S$ for all $e\in E$.
  In consequence, the problems are approximable within a constant if the number of scenarios  $K$ is assumed to be bounded
  ($K$~is bounded by a constant). 
 However,  up to now the existence of an approximation algorithm  with a constant performance ratio for the unbounded case has been an open question. 
In this paper we address this question and show that the minmax (regret) versions of \textsc{Shortest Path}, \textsc{Minimum Assignment} and \textsc{Minimum s-t Cut} are not approximable within $\log^{1-\epsilon} K$ for any $\epsilon>0$ unless NP $\subseteq$ DTIME$(n^{\mathrm{poly} \log n})$.
Here and subsequently $n$ denotes the length of the input.
The last inclusion is widely believed to be  untrue. 
We also show that \textsc{Minmax (Regret) Minimum Spanning Tree} is not approximable within $(2-\epsilon)$ for any $\epsilon>0$ unless $P=NP$.
Moreover, all the  negative results remain true even for a  
class of graphs with a very simple structure.
We can thus conclude that the discrete scenario representation of uncertainty leads to problems that are more complex to solve than the interval one. Recall that for the interval scenario case, if~$\mathcal{P}$ is polynomially solvable, then   \textsc{Minmax Regret}~$\mathcal{P}$  is approximable within~2~\cite{KA06}.

\section{The approximability of minmax (regret) network optimization problems}

In this section, we present the main results of
this paper. 
Namely,
we give a
negative answer to the question about
the existence of approximation algorithms 
with a constant performance ratio 
for the minmax 
(regret) versions of \textsc{Shortest Path}, \textsc{Minimum Assignment} and \textsc{Minimum s-t Cut},
when the number of scenarios $K$ is unbounded.
We increase the gaps obtained in~\cite{AI05c,AI07}  and prove that
the problems of interest are hard to approximate   within a ratio of
$\log^{1-\epsilon} K$ for any $\epsilon>0$.
We first consider the minmax (regret) versions of \textsc{Shortest Path}. In this problem set $\Phi$ consists of all paths between two distinguished nodes $s$ and $t$ of~$G$. 
\begin{thm}
\label{thm1}
		The \textsc{Minmax (Regret) Shortest Path} problem is not approximable within $\log^{1-\epsilon} K$ for any $\epsilon>0$, unless NP $\subseteq$ DTIME$(n^{\mathrm{poly} \log n})$ even for edge series-parallel directed or undirected graphs.
\end{thm}
\begin{proof}
Consider  the \textsc{3-SAT} problem, in which we are given  a set $\mathscr{U}=\{x_1,\dots,x_n\}$ of Boolean variables and a collection $\mathscr{C}=\{C_1,\dots,C_m\}$ of clauses, where every clause in $\mathscr{C}$  has exactly 
three distinct literals. We ask if there is an assignment to $\mathscr{U}$ that 
satisfies all clauses in $\mathscr{C}$. This problem is known to be strongly NP-complete~\cite{GR79}.

Given an instance of \textsc{3-SAT}, we construct the corresponding instance of \textsc{Minmax Shortest Path} as follows:
we associate with each clause $C_i=(l_i^1\vee l_i^2 \vee l_i^3)$
a digraph $G_i$ composed of~5 nodes: 
$s_i,v_1^i,v_2^i,v_3^i,t_i$ and 6 arcs: the arcs
$(s_i,v_1^i)$, $(s_i,v_2^i)$, $(s_i,v_3^i)$
correspond to literals in $C_i$ (\emph{literal arcs}),
the arcs $(v_1^i,t_i)$, $(v_2^i,t_i)$, $(v_3^i,t_i)$
have costs equal to~0 under every scenario 
(\emph{dummy arcs});  in order to construct digraph $G$,
 we connect all digraphs 
 $G_1,\ldots,G_m$ by dummy arcs $(t_i,s_{i+1})$ for $i=1,\ldots m-1$;  
we finish the construction of $G$ by setting $s=s_1$ and $t=t_m$. We now form scenario set $\Gamma$ as follows.
For every pair of arcs of $G$, $(s_i,v_j^i)$ and $(s_q,v_r^q)$,
that correspond to contradictory
 literals  $l_i^j$ and $l_q^r$, i.e. $l_i^j=\sim l_q^r$, we create scenario $S$ such that under this scenario the costs of the 
 arcs $(s_i,v_j^i)$ and $(s_q,v_r^q)$ are set to 1 and the costs of all the remaining arcs are set to~0. An example of the reduction is shown in Figure~\ref{fig2}.
 Notice that the resulting graph $G$ has a simple series-parallel topology (see for instance~\cite{VA82} for a description of this class of graphs).
 
\begin{figure}[ht]
			\psfrag{v11}{\tiny $v^{1}_1$}
			\psfrag{v12}{\tiny $v^{1}_2$}
			\psfrag{v13}{\tiny $v^{1}_3$}	
						
			\psfrag{v21}{\tiny $v^{2}_1$}
			\psfrag{v22}{\tiny $v^{2}_2$}
			\psfrag{v23}{\tiny $v^{2}_3$}
			
			\psfrag{v31}{\tiny $v^{3}_1$}
			\psfrag{v32}{\tiny $v^{3}_2$}
			\psfrag{v33}{\tiny $v^{3}_3$}

			\psfrag{s}{\tiny $s$}
			\psfrag{t}{\tiny $t$}
			\psfrag{s2}{\tiny $s_2$}
			\psfrag{s3}{\tiny $s_3$}
			\psfrag{t1}{\tiny $t_1$}
			\psfrag{t2}{\tiny $t_2$}
			\psfrag{x1}{\tiny $x_1$}
			\psfrag{nx2}{\tiny $\sim x_2$}
			\psfrag{nx3}{\tiny $\sim x_3$}	

			\psfrag{nx1}{\tiny $\sim x_1$}
			\psfrag{x2}{\tiny $x_2$}
			\psfrag{x3}{\tiny $x_3$}

\psfrag{X}{			
    \begin{footnotesize}
    \renewcommand{\arraystretch}{0.9}
    \setlength\tabcolsep{3pt}			
    \begin{tabular}{c|cccccc|}
                            &$S_1$&$S_2$&$S_3$&$S_4$&$S_5$&$S_6$\\
      \hline
    $(s,v^{1}_{1})$&\textbf{1}&0&0&0&0&0\\ 
    $(s,v^{1}_{2})$&0&\textbf{1}&\textbf{1}&0&0&0\\ 
    $(s,v^{1}_{3})$&0&0&0&\textbf{1}&\textbf{1}&0\\
    \hline
    $(s_2,v^{2}_{1})$&\textbf{1}&0&0&0&0&\textbf{1}\\
    $(s_2,v^{2}_{2})$&0&\textbf{1}&0&0&0&0\\
    $(s_2,v^{2}_{3})$&0&0&0&\textbf{1}&0&0\\
    \hline
    $(s_3,v^{3}_{1})$&0&0&0&0&0&\textbf{1}\\
    $(s_3,v^{3}_{2})$&0&0&\textbf{1}&0&0&0\\
    $(s_3,v^{3}_{3})$&0&0&0&0&\textbf{1}&0\\
    \end{tabular}
    \end{footnotesize}
}	
\psfrag{C}{\footnotesize
$\mathscr{C}=\{(x_1\vee\sim x_2\vee\sim x_3),
(\sim x_1\vee x_2 \vee x_3),(x_1 \vee x_2 \vee x_3)\}$}
						
      \includegraphics{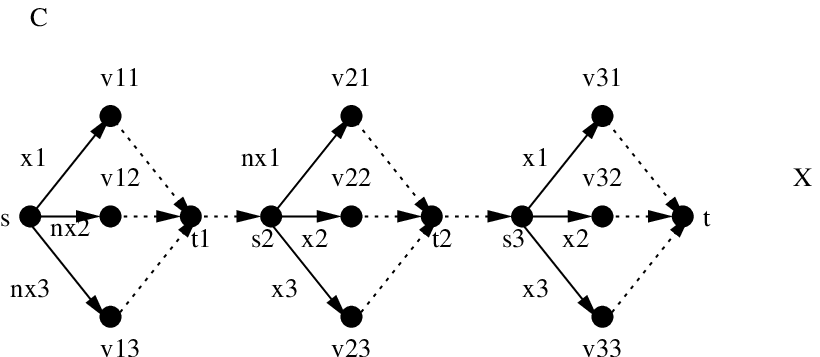}
\caption{\small An example of the reduction. 
All the dummy arcs (the dotted arcs) have costs equal to  0 under all scenarios and they are not listed in the table.} \label{fig2}
\end{figure}

It is easily verified that
if the answer to \textsc{3-SAT} is `Yes', then there is a path $P$ in $G$ that does not use arcs corresponding to contradictory literals. So, $\max_{S\in \Gamma} F(P,S)\leq 1$. On the other hand, if the answer is `No', then all paths in $G$ must use at least two arcs corresponding to contradictory literals and  $\max_{S\in \Gamma} F(P,S)\geq 2$. This yields a gap of~2 and  the \textsc{Minmax Shortest Path} problem is not approximable within~$(2-\epsilon)$ for any $\epsilon>0$ unless P=NP. 

We now show that the gap of~2 can be increased by applying an iterative construction that gradually increases the gap. A similar technique  was applied 
to the problem of minimizing the number of unsatisfied linear equations~\cite{AR97}  and
 to the problem of minimizing the number of nonzero variables
in linear systems~\cite{AM98}.

Let us transform the resulting graph $G=(V,E)$ into $G^{(1)}$ by replacing every  arc in $G$ by the whole graph $G$.  
We now associate scenario set $\Gamma^{(1)}$ with $G^{(1)}$ as follows. 
Initially,  $\Gamma$ has $K$ scenarios.
For every scenario $S\in \Gamma$ in graph $G$,
we create $K^2$ scenarios 
so that two values of~1 in $S$ are replaced
by all pairs of scenarios $S_i\in \Gamma$ and $S_j\in \Gamma$, $i,j=1,\ldots,K$.
In other words,  two values of~1 in $S$ are replaced
by two matrices $\mathbf{S}^{(1)}_1$ and $\mathbf{S}^{(1)}_2$
of the size $|E|\times K^2$, respectively,
where the columns of matrix
\[
\left(
\renewcommand{\arraycolsep}{2pt}
\begin{array}{c}
\mathbf{S}^{(1)}_1\\
\mathbf{S}^{(1)}_2
\end{array}
\right)
 = 
\left(
    \renewcommand{\arraycolsep}{2pt}			
\begin{array}{ccccccccccccc}
	S_1 &S_1&\ldots &S_1   &S_2 &S_2&\ldots &S_2&\ldots    &S_K&S_K&\ldots&S_K\\
	S_1 & S_2& \ldots & S_K&S_1 & S_2& \ldots & S_K&\ldots&S_1&S_2&\ldots&S_K
\end{array}
\right),\; S_i\in \Gamma, i=1,\ldots,K,
\]
 are  the Cartesian product $\Gamma\times \Gamma$. Furthermore,
  every value of~0 in $S$ is replaced by matrix $\mathbf{O}^{(1)}$ of the size $|E|\times K^2$ with all elements equal to zero. 
 The resulting instance is graph $G^{(1)}$ with $|E|^2$ edges and
 $K^3$ scenarios.
 A sample construction of $\Gamma^{(1)}$ is shown in Figure~\ref{fig3}.
Now, if the answer to \textsc{3-SAT} is `Yes', then there is a path $P$ in  graph $G^{(1)}$ such that $\max_{S\in \Gamma^{(1)}} F(P,S)\leq 1$ and $\max_{S\in \Gamma^{(1)}} F(P,S)\geq 4$ otherwise. We thus get a gap of~4. 

\begin{figure}[ht]
\begin{center}
\begin{small}
     \setlength\tabcolsep{3pt}
\begin{tabular}{c|cccccc|}
$G^{(1)}$&\multicolumn{6}{|c|}{$\Gamma^{(1)}$}\\ \hline
    $G_{(s,v^{1}_{1})}$&$\mathbf{S}^{(1)}_1$&$\mathbf{O}^{(1)}$&$\mathbf{O}^{(1)}$&$\mathbf{O}^{(1)}$&$\mathbf{O}^{(1)}$&$\mathbf{O}^{(1)}$\\ 
    $G_{(s,v^{1}_{2})}$&$\mathbf{O}^{(1)}$&$\mathbf{S}^{(1)}_1$&$\mathbf{S}^{(1)}_1$&$\mathbf{O}^{(1)}$&$\mathbf{O}^{(1)}$&$\mathbf{O}^{(1)}$\\ 
    $G_{(s,v^{1}_{3})}$&$\mathbf{O}^{(1)}$&$\mathbf{O}^{(1)}$&$\mathbf{O}^{(1)}$&$\mathbf{S}^{(1)}_1$&$\mathbf{S}^{(1)}_1$&$\mathbf{O}^{(1)}$\\
    \hline
    $G_{(s_2,v^{2}_{1})}$&$\mathbf{S}^{(1)}_2$&$\mathbf{O}^{(1)}$&$\mathbf{O}^{(1)}$&$\mathbf{O}^{(1)}$&$\mathbf{O}^{(1)}$&$\mathbf{S}^{(1)}_1$\\
    $G_{(s_2,v^{2}_{2})}$&$\mathbf{O}^{(1)}$&$\mathbf{S}^{(1)}_2$&$\mathbf{O}^{(1)}$&$\mathbf{O}^{(1)}$&$\mathbf{O}^{(1)}$&$\mathbf{O}^{(1)}$\\
    $G_{(s_2,v^{2}_{3})}$&$\mathbf{O}^{(1)}$&$\mathbf{O}^{(1)}$&$\mathbf{O}^{(1)}$&$\mathbf{S}^{(1)}_2$&$\mathbf{O}^{(1)}$&$\mathbf{O}^{(1)}$\\
    \hline
    $G_{(s_3,v^{3}_{1})}$&$\mathbf{O}^{(1)}$&$\mathbf{O}^{(1)}$&$\mathbf{O}^{(1)}$&$\mathbf{O}^{(1)}$&$\mathbf{O}^{(1)}$&$\mathbf{S}^{(1)}_2$\\
    $G_{(s_3,v^{3}_{2})}$&$\mathbf{O}^{(1)}$&$\mathbf{O}^{(1)}$&$\mathbf{S}^{(1)}_2$&$\mathbf{O}^{(1)}$&$\mathbf{O}^{(1)}$&$\mathbf{O}^{(1)}$\\
    $G_{(s_3,v^{3}_{3})}$&$\mathbf{O}^{(1)}$&$\mathbf{O}^{(1)}$&$\mathbf{O}^{(1)}$&$\mathbf{O}^{(1)}$&$\mathbf{S}^{(1)}_2$&$\mathbf{O}^{(1)}$\\
    \end{tabular}
\;\;\;
\begin{tabular}{c|cccccc|}
$G^{(2)}$&\multicolumn{6}{|c|}{$\Gamma^{(2)}$}\\ \hline
    $G^{(1)}_{(s,v^{1}_{1})}$&$\mathbf{S}^{(2)}_1$&$\mathbf{O}^{(2)}$&$\mathbf{O}^{(2)}$&$\mathbf{O}^{(2)}$&$\mathbf{O}^{(2)}$&$\mathbf{O}^{(2)}$\\ 
    $G^{(1)}_{(s,v^{1}_{2})}$&$\mathbf{O}^{(2)}$&$\mathbf{S}^{(2)}_1$&$\mathbf{S}^{(2)}_1$&$\mathbf{O}^{(2)}$&$\mathbf{O}^{(2)}$&$\mathbf{O}^{(2)}$\\ 
    $G^{(1)}_{(s,v^{1}_{3})}$&$\mathbf{O}^{(2)}$&$\mathbf{O}^{(2)}$&$\mathbf{O}^{(2)}$&$\mathbf{S}^{(2)}_1$&$\mathbf{S}^{(2)}_1$&$\mathbf{O}^{(2)}$\\
    \hline
    $G^{(1)}_{(s_2,v^{2}_{1})}$&$\mathbf{S}^{(2)}_2$&$\mathbf{O}^{(2)}$&$\mathbf{O}^{(2)}$&$\mathbf{O}^{(2)}$&$\mathbf{O}^{(2)}$&$\mathbf{S}^{(2)}_1$\\
    $G^{(1)}_{(s_2,v^{2}_{2})}$&$\mathbf{O}^{(2)}$&$\mathbf{S}^{(2)}_2$&$\mathbf{O}^{(2)}$&$\mathbf{O}^{(2)}$&$\mathbf{O}^{(2)}$&$\mathbf{O}^{(2)}$\\
    $G_{(s_2,v^{2}_{3})}$&$\mathbf{O}^{(2)}$&$\mathbf{O}^{(2)}$&$\mathbf{O}^{(2)}$&$\mathbf{S}^{(2)}_2$&$\mathbf{O}^{(2)}$&$\mathbf{O}^{(2)}$\\
    \hline
    $G^{(1)}_{(s_3,v^{3}_{1})}$&$\mathbf{O}^{(2)}$&$\mathbf{O}^{(2)}$&$\mathbf{O}^{(2)}$&$\mathbf{O}^{(2)}$&$\mathbf{O}^{(2)}$&$\mathbf{S}^{(2)}_2$\\
    $G^{(1)}_{(s_3,v^{3}_{2})}$&$\mathbf{O}^{(2)}$&$\mathbf{O}^{(2)}$&$\mathbf{S}^{(2)}_2$&$\mathbf{O}^{(2)}$&$\mathbf{O}^{(2)}$&$\mathbf{O}^{(2)}$\\
    $G^{(1)}_{(s_3,v^{3}_{3})}$&$\mathbf{O}^{(2)}$&$\mathbf{O}^{(2)}$&$\mathbf{O}^{(2)}$&$\mathbf{O}^{(2)}$&$\mathbf{S}^{(2)}_2$&$\mathbf{O}^{(2)}$\\
    \end{tabular}    
   \end{small}
   \end{center}
  \caption{The construction of $\Gamma^{(1)}$ and $\Gamma^{(2)}$ for the sample problem shown in Figure~\ref{fig2}. 
$G_{(s_i,v^{i}_{j})}$~and $G^{(1)}_{(s_i,v^{i}_{j})}$ are the graphs
$G$ and $G^{(1)}$, respectively,   
inserted in place of~$(s_i,v^{i}_{j})$ in $G$. The graphs corresponding to the dummy arcs are not shown.} \label{fig3}
  \end{figure}

We can now  repeat the construction
to obtain an instance with a gap of~8. Namely,
we construct $G^{(2)}$   by replacing every  arc
in graph $G$ by the whole graph $G^{(1)}$.
Then we form scenario set $\Gamma^{(2)}$ in the following way.
For every scenario $S\in \Gamma$ in graph $G$,
we create $(K^3)^2$ scenarios 
such that two values of~1 in $S$ are replaced
by two matrices $\mathbf{S}^{(2)}_1$ and $\mathbf{S}^{(2)}_2$
of the size $|E|^2\times (K^3)^2$, respectively,
where all columns of matrix 
\[
\left(
\renewcommand{\arraycolsep}{0.7pt}
\begin{array}{c}
\mathbf{S}^{(2)}_1\\
\mathbf{S}^{(2)}_2
\end{array}
\right)
 = 
\left(
    \renewcommand{\arraycolsep}{0.9pt}			
\begin{array}{ccccccccccccc}
	S^{(1)}_1 &S^{(1)}_1&\ldots &S^{(1)}_1   &S^{(1)}_2 &S^{(1)}_2&\ldots &S^{(1)}_2&\ldots    &S^{(1)}_{K^3}&S^{(1)}_{K^3}&\ldots&S^{(1)}_{K^3}\\
	S^{(1)}_1 & S^{(1)}_2& \ldots & S^{(1)}_{K^3}&S^{(1)}_1 & S^{(1)}_2& \ldots & S^{(1)}_{K^3}&\ldots&S^{(1)}_1&S^{(1)}_2&\ldots&S^{(1)}_{K^3}
\end{array}
\right),\mbox{ } S^{(1)}_i\in \Gamma^{(1)}, i=1,\ldots,{K^3},
\]
 are  the Cartesian product $\Gamma^{(1)}\times \Gamma^{(1)}$,
 and every value of~0 in $S$ is replaced by matrix $\mathbf{O}^{(2)}$ of the size $|E|^2\times (K^3)^2$ with all elements equal to zero
 (see Figure~\ref{fig3}).
By repeating the above construction $t$ times, we get   graph $G^{(t)}$ with $|E|^{t+1}$ edges together with scenario set $\Gamma^{(t)}$ containing $K^{2^{t+1}-1}$ scenarios that yield a total gap of 
$2^{t+1}$. 
Let $t=\log \log^\beta n$ for some fixed $\beta>0$, where $n$ is the number of variables in the instance of  \textsc{3-SAT}.
 Now graph $G^{(t)}$ has $|E|^{\log \log^\beta n +1}$ edges 
 and $K^{2\log^{\beta} n -1}$ scenarios. 
 Let $K'=K^{2\log^{\beta} n -1}$.
 Since $|E|$ and $K$ are bounded by a polynomial in $n$, graph $G^{(t)}$ 
 together with scenario set $\Gamma^{(t)}$ 
 can be constructed in $O(n^{\mathrm{poly} \log n})$ time. 
 The resulting instance of \textsc{Minmax Shortest Path} has a gap of $2\log^{\beta} n$. 
 Since $K'=K^{2\log^{\beta} n -1}$ and $K=O(n^c)$ for some constant $c$, we get $K'=2^{\log K (2\log^{\beta} n-1)}=2^{O(\log^{\beta+1} n)}$. So, $\log K'=O(\log^{\beta+1} n)$ and the gap is $2\log^{\beta} n=O(\log^{\frac{\beta}{\beta+1}} K')$. 

Assume, on the contrary, that a polynomial time algorithm approximates the \textsc{Minmax Shortest Path} problem within a factor $\log^{1-\epsilon} K$ for any $\epsilon>1-\frac{\beta}{\beta+1}$. Applying this algorithm to the resulting graph $G^{(t)}$ with scenario set $\Gamma^{(t)}$ containing $K'$ scenarios, we could decide the \textsc{3-SAT} problem in $O(n^{\mathrm{poly} \log n})$ time. But this would imply NP $\subseteq$ DTIME$(n^{\mathrm{poly} \log n})$,  a contradiction.

In order to prove the result for 
\textsc{Minmax Regret Shortest Path}, we use  exactly the same 
graph $G^{(t)}$ with scenario set $\Gamma^{(t)}$.
A proof goes without any modifications.
It follows from the fact that under every scenario $S\in \Gamma^{(t)}$
there is a path $P^{*}$ in $G^{(t)}$ such that $F(P^{*},S)=0$,
and thus $F^{*}(S)=0$.
In consequence the minmax and minmax regret criteria lead to  optimal solutions with the same total costs in the resulting 
instances. Furthermore,
if arc directions are ignored, then the results hold for 
the undirected graphs as well.
\end{proof}
It is worth pointing out that Theorem~\ref{thm1}
holds for the graphs having a very simple series-parallel
structure. The class of series-parallel graphs
is a subclass of various classes of graphs and its description can be found for instance in~\cite{VA82}. Recall also that in the interval scenario case  the \textsc{Minmax Regret Shortest Path} problem for edge series-parallel graphs admits a
fully polynomial time approximation scheme~\cite{KA07}.

In the \textsc{Minimum Assignment} problem we assume that $G$ is a bipartite graph and $\Phi$ consists of all perfect matchings in $G$. 
The following corollary is 
a  consequence of Theorem~\ref{thm1}:
\begin{cor}
The \textsc{Minmax (Regret) Minimum Assignment} problem is not approximable within $\log^{1-\epsilon} K$ for any $\epsilon>0$, unless NP $\subseteq$ DTIME$(n^{\mathrm{poly} \log n})$.
\end{cor}
\begin{proof}
In~\cite{AI05} a  
cost preserving reduction from \textsc{Minmax (Regret) Shortest Path} to \textsc{Minmax (Regret) Minimum Assignment} has been proposed.
Therefore, we have exactly the same inapproximability
results for \textsc{Minmax (Regret) Minimum Assignment}
as for \textsc{Minmax (Regret) Shortest Path}.
\end{proof}

Recall that in the \textsc{Minimum Spanning Tree} problem $\Phi$ consists of all spanning trees of~$G$, that is all subsets of exactly $|V|-1$ edges that form acyclic subgraphs of $G$. The following result is true:
\begin{cor}
\label{cor1}
The \textsc{Minmax (Regret) Minimum Spanning Tree} problem is not approximable within $(2-\epsilon)$ for any $\epsilon>0$, unless P=NP even for edge series-parallel graphs.
\end{cor}
\begin{proof}
It is enough to observe that an optimal minmax (regret) path in the first graph $G$ from the proof of Theorem~\ref{thm1} can be transformed to an optimal minmax (regret) spanning tree and vice versa by adding or removing a number of dummy edges. 
Since the dummy edges have costs equal to~0 under all scenarios, this transformation is cost preserving. We get a gap of~2 and the theorem follows.
\end{proof}
Notice that Corollary~\ref{cor1} strengthens the results obtained in~\cite{AI05c}. However, we are not able to show here that \textsc{Minmax (Regret) Minimum Spanning Tree} is not approximable within a constant factor. The reduction, which is valid  for the first graph $G$, is not true for the subsequent graphs $G^{(t)}$. In other words, it is not possible to transform a path in $G^{(t)}$ into a spanning tree by simply adding dummy edges. The question whether \textsc{Minmax (Regret) Minimum Spanning Tree} is approximable within a constant remains open. It is also open for a more general class of minmax (regret) matroidal problems, where $\Phi$ consists of all bases of a given matroid~\cite{law76}. 

Finally, let us consider the \textsc{Minimum s-t Cut} problem. In this problem we distinguish two nodes $s$ and $t$ in $G$ and $\Phi$ consists of all $s-t$-cuts in $G$, that is the subset of edges whose removal disconnects $s$ and $t$. 
\begin{thm}
The
 \textsc{Minmax (Regret) s-t Cut}  problem
  is not approximable within $\log^{1-\epsilon} K$ for any $\epsilon>0$, unless NP $\subseteq$ DTIME$(n^{\mathrm{poly} \log n})$ even for edge series-parallel graphs.
\end{thm}
\begin{proof}
As in the proof of Theorem~\ref{thm1}, 
we  use a reduction  from \textsc{3-SAT}. 
For a given instance of \textsc{3-SAT}, we construct the corresponding instance of \textsc{Minmax s-t Cut} as follows: 
for each clause $C_i=(l_i^1\vee l_i^2 \vee l_i^3)$ in~$\mathscr{C}$,
we create three edges of the form $\{s,v_1^i\}$, $\{v_1^i,v_2^i\}$ and $\{v_2^i,t\}$ that correspond to the literals in $C_i$.
Observe that the resulting graph $G$ is composed of exactly $m$ disjoint 
$s-t$ paths and it has a series-parallel topology.
For every pair of edges, that correspond to \emph{contradictory} literals  $l_i^j$ and $l_q^r$, we form scenario $S$ such that under this scenario the costs of edges that correspond to $l_i^j$ and $l_q^r$ are set to~1 and the costs of all the remaining edges are set to 0.	 An example of the reduction is shown in Figure~\ref{fig4}.
	\begin{figure}[ht]
\begin{center}
			\psfrag{C}{\footnotesize 
			$\mathscr{C}=\{(x_1\vee \sim x_2\vee \sim x_3),(\sim x_1\vee  x_2\vee x_3),(x_1\vee  x_2\vee  x_3)\}$}
			\psfrag{v11}{\tiny $v^{1}_1$}
			\psfrag{v12}{\tiny $v^{1}_2$}

			\psfrag{v21}{\tiny $v^{2}_1$}
			\psfrag{v22}{\tiny $v^{2}_2$}

			\psfrag{v31}{\tiny $v^{3}_1$}
			\psfrag{v32}{\tiny $v^{3}_2$}

			\psfrag{s}{\tiny $s$}
			\psfrag{t}{\tiny $t$}
			\psfrag{x1}{\tiny $x_1$}
			\psfrag{nx2}{\tiny $\sim x_2$}
			\psfrag{nx3}{\tiny $\sim x_3$}	

			\psfrag{nx1}{\tiny $\sim x_1$}
			\psfrag{x2}{\tiny $x_2$}
			\psfrag{x3}{\tiny $x_3$}	
			
			\psfrag{X}{			
    \begin{footnotesize}
    \renewcommand{\arraystretch}{0.9}
    \setlength\tabcolsep{3pt}			
    \begin{tabular}{c|cccccc|}
                            &$S_1$&$S_2$&$S_3$&$S_4$&$S_5$&$S_6$\\
      \hline
    $\{s,v^{1}_{1}\}$&\textbf{1}&0&0&0&0&0\\ 
    $\{v^{1}_{1},v^{1}_{2}\}$&0&\textbf{1}&\textbf{1}&0&0&0\\ 
    $\{v^{1}_{2},t\}$&0&0&0&\textbf{1}&\textbf{1}&0\\
    \hline
    $\{s,v^{2}_{1}\}$&\textbf{1}&0&0&0&0&\textbf{1}\\
    $\{v^2_1,v^{2}_{2}\}$&0&\textbf{1}&0&0&0&0\\
    $\{v^{2}_{2},t\}$&0&0&0&\textbf{1}&0&0\\
    \hline
    $\{s,v^{3}_{1}\}$&0&0&0&0&0&\textbf{1}\\
    $\{v^{3}_{1},v^3_2\}$&0&0&\textbf{1}&0&0&0\\
    $\{v^{3}_{2},t\}$&0&0&0&0&\textbf{1}&0\\
    \end{tabular}
    \end{footnotesize}
}	
					
      \includegraphics{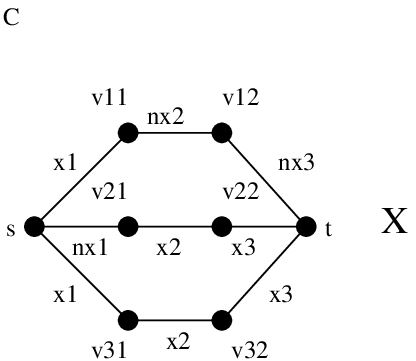}
\end{center}
\caption{\small An example of the reduction.} \label{fig4}
\end{figure}

It is easy to
check that the answer to \textsc{3-SAT} is `Yes' if there is a cut $C$ in $G$ such that $\max_{S\in \Gamma} F(C,S)\leq 1$ and $\max_{S\in \Gamma} F(C,S)\geq 2$ otherwise. We thus get a gap of~2 and the \textsc{Minmax s-t Cut} problem is not approximable within~$(2-\epsilon)$ for any $\epsilon>0$ unless P=NP. The rest of the proof is the same as the one of Theorem~\ref{thm1}.
\end{proof}

\section{Conclusions}

In this paper, we have given a
negative answer to the question about
the existence of  approximation algorithms
with a constant performance ratio 
for  the minmax and minmax regret versions 
of  \textsc{Shortest Path}, 
\textsc{Minimum Assignment} and \textsc{Minimum s-t Cut}
under discrete scenario representation of uncertainty.
Namely, we have shown that the considered problems
are hard to approximate   within a ratio of
$\log^{1-\epsilon} K$ for any $\epsilon>0$
if the number of scenarios $K$ is unbounded.
The question whether the performance ratio of $K$ is
the best possible for these problems remains open
and  it is the subject of future research. We have also strengthen the known results for the \textsc{Minmax (Regret) Minimum Spanning Tree} problem. For this problem, however, it may be still possible to design an approximation algorithm with a constant performance ratio.


\subsubsection*{Acknowledgements}
This work was 
partially supported by Polish Committee for Scientific Research, grant
N N111 1464 33.

\end{document}